\documentclass[preprint,12pt]{elsarticle}

\usepackage[utf8]{inputenc}
 \usepackage[english]{babel}

\usepackage{graphics}
\usepackage{color}
\usepackage{url}

\usepackage{amssymb,amsfonts,amsmath,amsthm}
\usepackage{subfigure}

\newtheorem{theorem}{Theorem}[section]

\newtheorem{proposition}[theorem]{Proposition}

\newtheorem{remark}{Remark}
\newtheorem{definition}[theorem]{Definition}

\journal{}

\begin{document}
\begin{frontmatter}

\title{A Metapopulation Model for Chikungunya Including Populations Mobility on a Large-Scale Network}

\author{Djamila Moulay\corref{cor1}}
\ead{djamila.moulay@univ-lehavre.fr}
\ead[url]{http://djamila-moulay.org}
\cortext[cor1]{Corresponding author. LMAH, Universit\'e du Havre, 25 rue Philippe Lebon, BP540, 76058 Le Havre Cedex}
\address{LMAH, Le Havre University, France}

\author{Yoann Pign\'e}
\ead{yoann.pigne@univ-lehavre.fr}
\ead[url]{http://pigne.org}
\address{LITIS, Le Havre University, France}

\begin{abstract}
In this work we study the influence of populations mobility on the spread of a vector-borne disease. We focus on the chikungunya epidemic event that occurred in 2005-2006 on the R\'eunion Island, Indian Ocean, France, and validate our models with real epidemic data from the event. We propose a metapopulation model to represent both a high-resolution patch model of the island with realistic population densities and also mobility models for humans (based on real-motion data) and mosquitoes. In this metapopulation network, two models are coupled: one for the dynamics of the mosquito population and one for the transmission of the disease. A high-resolution numerical model is created out from real geographical, demographical and mobility data. The Island is modeled with an 18,000-nodes metapopulation network. Numerical results show the impact of the geographical environment and populations' mobility on the spread of the disease. The model is finally validated against real epidemic data from the R\'eunion event.
\end{abstract}


\end{frontmatter}
\section{Introduction}

Many factors have influence the emergence and re-emergence of vector-borne diseases \cite{gratz_emerging_1999, Sutherst2004}. Brutal changes in natural habitats so as massive or recurrent population migrations tend to speed up the spread of vector-borne diseases. The spatiotemporal evolution of such diseases is becoming a key issue for epidemiologists. To this purpose, models that consider the spatial distribution of a natural environment are of great interest. In this paper we are interested in the spatial spread of  a vector-borne disease under the effect of human and vector mobility.
Indeed, human migration is one of the factor that have influenced the re-emergence of several diseases \cite{national2006The, Martens2000}. 
The modeling of geographical environments and populations mobilities are becoming mandatory in this context. 
 There are several approaches to describe such spread. One of the typical approaches, introducing spatial spread variation in epidemic models, involves the use of partial differential equations \cite{Bailey1980,Murray2002a,Fitzgibbon2003}. 
However, in the case of human mobility these approaches may not be appropriate.   
The theory of "metapopulation", first introduced in 1969 \cite{Levins1969}, in the field of ecology, allows such modeling. 
Several researches have been devoted to the study of disease spread in heterogenous environments \cite{dushoff_effects_1995, Wang2004,lloyd_spatial_1996}.
In \cite{Colizza2008,Cosner2009}, authors study the influence of human dispersal among $n$ patches in the dynamics of disease spread. In \cite{Balcan2010}, authors, to study the spread of seasonal influenza, focus on air displacements and model the environment with  a network where nodes are airports and edges represent flights. In \cite{LonginiJr.1988} and \cite{Rvachev1985}, authors are also interested in the case of  influenza and mobility in terms of long journeys.
In \cite{Van2003,Arino2009,ArinoJune2005}, the authors propose a virus spread model based on this theory. In \cite{Menach2005,Smith2004}, the mobility of mosquitoes has been modeled to study the spatiotemporal dynamics of malaria. Other studies focus on direct-transmission diseases like \cite{ArinoJune2005,Van2003}. In \cite{Balcan2010}, influenza in the case of long trips (aircraft flights) is tackled. In  \cite{Auger2008}, the authors rely on the Ross-Macdonald model \cite{Ross1910,L.W.Hackett09011958}, to consider human mobility. 

Human mobility is usually considered in epidemic problems. Also, models that consider the dynamics of the vector population are numerous. However, to our knowledge, no  model considers to couple population dynamics with populations mobility as we propose here.  Moreover our approach promotes the consideration of vectors's mobility since the resolution of our model is greater than usual models.

In this work we are proposing to couple two models (published in \cite{Moulay2010}): a mosquito dynamics model (growth and evolution of the population) with a transmission virus model between two populations (humans and mosquitoes). The dynamics of the mosquito population is described by a stage structured  model based on its biological life cycle. The different compartment states are egg, larva, pupa and adult. 
The disease transmission model is, for the human population, a SIR (Susceptible, Infected, Recovered) model. For adult mosquitoes, the transmission model is a SI (Susceptible, Infected) model. 

Those two models are formalized using the metapopulation theory. It considers a network where nodes represent real habitats of the environment. In each of these nodes, transmission and population dynamics models appear and are coupled with neighbor nodes. Links in the network represent both the local neighborhood of a node and  farther nodes that code for the mobility of humans. 

We focus on a real case of chikungunya epidemic with the 2005-2006 event that occurred on the R\'eunion Island, a French island in the Indian Ocean. The island is modeled with a network. Since we want that network to reflect the local population's density, we consider the road network of the island as a proxy to the human density, considering that each crossroad is a node of the network. Then the local population on each node is adjusted according to real data given by the French Institute for Statistics (INSEE). Finally the all island is modeled with a 18,000 nodes network. 

In \cite{Gonzalez:2008fk}, authors propose various distributions that match a data set of real human mobility patterns (cell phone probes). We rely on these distributions to create human mobility patterns for the population in our model. We assume that individuals only change disease status when they are on a node and not during displacements. See \cite{Cui2006} for a model with disease transmission during travel.

Results show the decisive influence of mobilities over the spread of the disease. Not only the human mobility but also the vector local interaction that play an important role at the considered scale. Results are then compared to real data epidemic  regarding the 2005-2006 event at the scale of the all R\'eunion Island and the model is validated. 

The remaining of this paper is organized as follows. Next section recalls original transmission and population dynamics models that this work is inspired  of. Section \ref{sec:metapopulation_model} introduces the original metapopulation model that is able to include the previous models in a network of patches, linked to represent populations mobility. Then, in Sect. \ref{sec:application}, according to the wish to validate the model against a real epidemic, a numerical implementation of the metapopulation model is constructed. This section details the construction of the metapopulation network in terms of environment and populations mobility. Section \ref{sec:results} presents various analysis performed on the model, showing the impact of mobilities. A validation is then proposed with a comparison with real epidemic data. Then a stochastic analysis of the model is proposed to study the robustness of the system. Finally Sect. \ref{sec:conclusion} concludes that paper. 

\section{Original Model} 
\label{sec:original_model}


We first recall the model proposed and studied in \cite{Moulay2010} describing the vector dynamics. The formulated model is a stage structured model based on the biological mosquito life cycle. The vector population is subdivided into several classes: the aquatic stagese consisting in eggs $E$ and larvae/pupae $L$, and then the adult stage $A$, representing adult females.
We assume first  that the number of eggs $b$, laid by females, is proportional to the number of females itself. Secondly,  the number of eggs and larvae is regulated by the effect of a carrying capacity $K_E$ and $K_L$, respectively.
Then we formulate a transmission virus model where the adult female stage $A$ is  subdivided into two epidemiological states: susceptible $S_m$ and infectious $I_m$.
It assumes that there is no vertical transmission of the virus, so that births from susceptible and infectious mosquitos occur into the egg stage with the same rate. 

%
The human population consists in three epidemiological states: susceptible (or non-immune) ${S}_H$, infectious ${I}_H$, and removed (or immune) ${R}_H$. 
It is assumed that there is no vertical transmission of the disease, so that all births occur into the susceptible human class, at the rate  $b_H > 0$.
Moreover we assume that the total human population $N_H$ remains constant, thus birth and death rates are equal.  
An infected human is infectious during $1/\gamma_h$ days, called the viremic period, and then becomes resistant or immune.

Forces of infections used in the model, which describe the rates of apparition of new infections, are standard and  modeled by the mass-action principle normalized by the total population of humans, given by  $ \beta_m {I}_H(t) S_m(t)/N_H $ and  $ \beta_H {I}_m(t){S}_H(t)/{N_H} $ where $\beta_m$ and $\beta_H$ are the transmission parameters.

This hypothesis are summed up in Fig. \ref{cycle_schema}.
\begin{figure}[htb]
\centering
\includegraphics[height=0.6\linewidth]{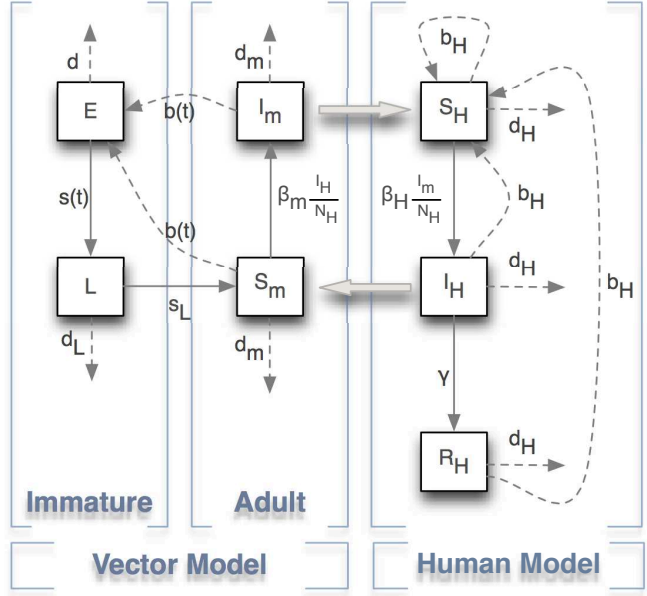}
\caption {Compartmental model for the dynamics of \textit{Aedes albopictus} mosquitos and the virus transmission to human population. 
 $b(t)=bA(t)\left(1-E(t)/K_E\right)$, $s(t)=sE(t)\left(1-L(t)/K_L\right)$ and the parameters of the model are given in table \ref{tableau1}.
}
 \label{cycle_schema}
 \end{figure}

\begin{table}

\caption{Parameters of the mosquito and transmission dynamics models. Most of the values were obtained from other publications in the field (\cite{Bacaer2007,Delatte2009,Dumont2010,Moulay2012}).} \label{tableau1} 
 \centering
 {\footnotesize
 \begin{tabular}{llr}
  \hline 
 Parameters & Description & Value\\
 \hline 
 $b$ & oviposition rate & $6.0$\\
 $K_E$&  egg carrying capacity & $1000$\\
 $K_L$ & larva carrying capacity & $K_E/2$\\
 $s$ & transfer rate $E\rightarrow L$ & $1/3$\\
 $s_L$ & transfer rate $L\rightarrow A$ & $1/10$\\
 $d$& egg mortality rate & $1/3$\\
 $d_L$& larva mortality rate & $1/3$\\
 $d_m$ & female adult mortality rate & $1/14$\\
 $b_H$ & human birth/death rate & $1/(78*365)$\\
 $\beta_H$& infection rate : vector to human  & See Sect. \ref{sub:parameters_analysis}\\ 
 $\beta_m$ & infection rate :  human to vector  & See Sect. \ref{sub:parameters_analysis}\\ 
 $\gamma_H$ & human recovery rate & $1/7$\\
 \hline\\
 \end{tabular}
 }
 \end{table}

Based on our model description (see Fig. \ref{cycle_schema}) and assumptions, we establish the following equations.
\begin{align}
  \label{transmission8}
  \left\{
    \begin{array}{rcl}
      \dfrac{dE}{dt}(t)&=&\displaystyle bA(t)\left(1-\dfrac{E(t)}{K_E}\right)-(s+d)E(t)\\
      \dfrac{dL}{dt}(t)&=&\displaystyle s E(t)\left(1-\dfrac{L(t)}{K_L}\right)-(s_L+d_L) L(t)\\
      \dfrac{dA}{dt}(t)&=&\displaystyle s_L L(t)-d_m A(t)\\
      \dfrac{d{S}_m}{dt}(t)&=&\displaystyle s_L L(t)-d_m {S}_m(t)-\beta_m \dfrac{{I}_H(t)}{N_H} {S}_m(t)\\
      \dfrac{d{I}_m}{dt}(t)&=&\displaystyle\beta_m \dfrac{{I}_H(t)}{N_H} {S}_m(t) -d_m{I}_m(t)\\
     \dfrac{d {S}_H}{dt}(t)&=&\displaystyle -\beta_H \dfrac{{I}_m(t)}{N_H} {S}_H(t)+b_H({S}_H(t)+{I}_H(t)+{R}_H(t))-d_{H}S_H(t)\\
     \dfrac{d {I}_H(t)}{dt}&=&\displaystyle\beta_H \dfrac{{I}_m(t)}{N_H} {S}_H(t)-\gamma {I}_H(t)-d_{H}{I}_H(t)\\
     \dfrac{d{R}_H(t)}{dt}&=&\displaystyle\gamma {I}_H(t)-d_{H}R_H(t)
    \end{array}
  \right.
\end{align}

The study of the model is detailed in \cite{Moulay2010}. For this model \eqref{transmission8}, the second basic reproduction number is given by 
$$R_0
   =  \dfrac{\beta_m \beta_H}{d_m (\gamma +b_H)}\dfrac{1}{N_H}\left(1-\dfrac{1}{r}\right)\dfrac{sK_Es_LK_L}
{d_m\big(sK_E+(s_L+d_L)K_L\big)},$$ where $\displaystyle r=({b}/(s+d))({s}/(s_L+d_L))({s_L}/{d_m}).$ 

The threshold $r$ governs the dynamics of mosquitoes. In this article we assume that $r>1$. This ensures the existence, persistence and global stability of a unique endemic equilibrium $(E^*,L^*,A^*)$, which corresponds to the biological and interesting case. The second reproduction number $R_0$ governs the dynamic of the transmission model.

\section{Metapopulation Model} 
\label{sec:metapopulation_model}

In \cite{Arino2009,ArinoJune2005,Van2003} J. Arino \emph{et al.} formulate a general system of differential equations allowing to describe human mobility. In this model, they identify each population by its origin and its present location. We rely on this model and extend it with the definition of a neighborhood in which humans and mosquitoes interact.

\subsection{Human mobility}

Based on the model given in \cite{Arino2009,ArinoJune2005,Van2003}, we proposed an extension of our previous model \cite{Moulay2010}, to describe the spread of a chikungunya under human and vector mobility on a large scale network.

Assume that the number of nodes in the network is $n$. A human population is identified due to two characteristics: the node from which he is originated, \textit{i.e.} its residence and the node where he is at time $t$. We assume that the total human population is constant,  \textit{i.e.} births and deaths occur with the same rate $b_H=d_H$. Moreover, we suppose that birth occur in the resident node while deaths take place in any node the human is present.


Let us denote by ${S_H}_{ij}(t), {I_H}_{ij}(t)$ and ${R_H}_{ij}(t)$ the susceptible, infected and removed human populations originated from node $i$ and present on node $j$ at time $t$.
We  denote by ${S_m}_i$ and ${I_H}_i$ the susceptible and infected mosquito present in node $i$. In a first assumption  mosquitoes mobility is neglected, which looks like realistic compared to the distance of humans displacements. So, resident mosquitoes from node $i$ are also present on this node.

The total number of susceptible, infected and removed human residents of node $i$, 
 is given respectively by 
${S_H}_i^r=\sum_{j=1}^{n} {S_H}_{ij}$, ${I_H}_i^r=\sum_{j=1}^{n} {I_H}_{ij}$ and  ${R_H}_i^r=\sum_{j=1}^{n} {R_H}_{ij}.$\\
The total number of susceptible, infected and removed humans present on node $i$,
is given respectively by
${S_H}_i^p=\sum_{j=1}^{n} {S_H}_{ji}$, ${I_H}_i^p=\sum_{j=1}^{n} {I_H}_{ji}$ and ${R_H}_i^p=\sum_{j=1}^{n} {R_H}_{ji}.$\\
The number of human residents on node $i$ is then equal to ${N_H}_i^r={S_H}_{i}^r+{I_H}_{i}^r+{R_H}_{i}^r$ and the human population present on node $i$ is ${N_H}_i^p= {S_H}_i^p+{I_H}_i^p+{R_H}_i^p$. \\

As in \cite{Sattenspiel1988} and \cite{Van2003}, we define the travel rate from node $i$ to node $j$ by $g_i m_{ji}$, where $g_i\geq 0$ corresponds to the per capita rate at which residents of node $i$ leave this node and a fraction  $m_{ji}\geq 0$ of them go to node $j$, with $m_{ii}=0$. Residents of node $i$, present on node $j$, then return to node $i$ with a per capita rate $r_{ij}\geq 0$, with $r_{ii}=0$. 
Displacements between two nodes are represented in Fig. \ref{fig:schema_deplacement}. 
\begin{figure}[ht!] \centering 
	\includegraphics[width=0.7\linewidth]{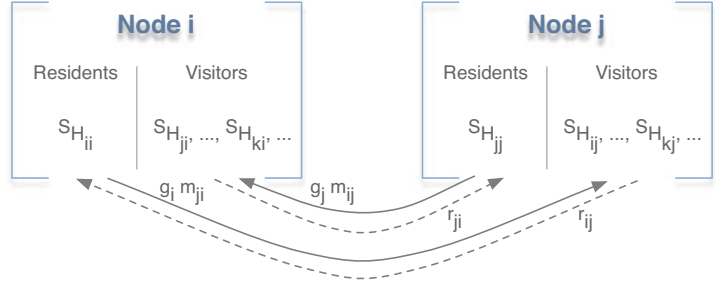}
	\caption{Human population mobility between nodes $i$ and $j$.}\label{fig:schema_deplacement} 
\end{figure}

The dynamics of human populations (susceptible, infected and removed) resident and present on node $i$ is given by:
\begin{eqnarray*}
\dfrac{d{S_H}_{ii}}{dt}&=&d_H( {N_H}_i^r- {S_H}_{ii})-g_i {S_H}_{ii}+\sum_{k=1}^n r_{ik} {S_H}_{ik}- {{\beta}_H}_i \dfrac{Im_{i}}{{N_H}_i^p} {S_H}_{ii}\\[-3mm]
  \dfrac{d{I_H}_{ii}}{dt}&=&- d_H {I_H}_{ii}-g_i {I_H}_{ii}+\sum_{k=1}^n r_{ik} {I_H}_{ik}+{{\beta}_H}_i \dfrac{Im_{i}}{{N_H}_i^p} {S_H}_{ii}- \gamma_H {I_H}_{ii}\\[-3mm]
 \dfrac{d{R_H}_{ii}}{dt}&=&- d_H {R_H}_{ii}+\gamma_H {I_H}_{ii} -g_i {R_H}_{ii}+\sum_{k=1}^n r_{ik} {R_H}_{ik}
\end{eqnarray*}

The dynamics of human populations (susceptible, infected and removed) resident on node $i$  and present on node $j$ is given by:
\begin{eqnarray*}
 \dfrac{d{S_H}_{ij}}{dt}&=& -d_H {S_H}_{ij}+g_i m_{ji} {S_H}_{ii}-r_{ij} {S_H}_{ij}- {{\beta}_H}_j \dfrac{Im_{j}}{{N_H}_j^p} {S_H}_{ij}\\[-2mm]
  \dfrac{d{I_H}_{ij}}{dt}&=&- d_H {I_H}_{ij}+g_i m_{ji} {I_H}_{ii}-r_{ij} {I_H}_{ij}+{{\beta}_H}_j \dfrac{Im_{j}}{{N_H}_j^p} {S_H}_{ij}- \gamma_H {I_H}_{ij}\\[-2mm]
 \dfrac{d{R_H}_{ij}}{dt}&=&- d_H {R_H}_{ij}+g_i m_{ji} {R_H}_{ii}- r_{ij} {R_H}_{ij}+\gamma_H {I_H}_{ij}.
\end{eqnarray*}

The dynamic of mosquito populations (eggs, larvae and pupae, adult females) given in \cite{Moulay2010}  in each node is:
 \begin{subequations}
  \label{eq_mous}
\begin{align}
&\dfrac{d{S_m}_{i}}{dt}= s_L L_i- {d_m}{S_m}_{i}-  {\beta_m}_i \dfrac{{S_m}_{i}}{{N_H}_i^p} {I_H}_{i}^p\label{eq_mous_a}\\[-1mm]
 &\dfrac{d{I_m}_{i}}{dt} ={\beta_m}_i \dfrac{{S_m}_{i}}{{N_H}_i^p} {I_H}_{i}^p-{d_m} {I_m}_{i}\label{eq_mous_b}\\[-1mm]
&\dfrac{dE_i}{dt}=b({S_m}_i(t)+{I_m}_i(t))\left(1-\dfrac{E_i(t)}{{K_E}_i}\right)-(s+d)E_i(t)\label{eq_mous_c}\\[-1mm]
&\dfrac{dL_i}{dt}=s E_i(t)\left(1-\dfrac{L_i(t)}{{K_L}_i}\right)-(s_L+d_L)L_i(t)\label{eq_mous_d},
\end{align}
\end{subequations}
where the dynamic of immature stages $E$ and $L$ is described by corresponding equations in system \eqref{transmission8}.


\begin{remark}
\begin{enumerate}[i.]
\item Note that the infection transmitted in node $i$ between susceptible mosquitoes and infected humans depends now on the human population present on this node ${N_H}_i^p$:
$\displaystyle\sum_{j=1}^n {\beta_m}_i \dfrac{{S_m}_{i}}{{N_H}_i^p} {I_H}_{ji}={\beta_m}_i \dfrac{{S_m}_{i}}{{N_H}_i^p} {I_H}_{i}^p.$
\item The distance between nodes is not explicitly taken into account, nevertheless it is implicitly included in the coefficient $m_{ji}$ and $r_{ij}$.
\item In the following we focus on daily displacements, humans who leave their resident node, obviously return to their resident node on a daily basis. Displacement matrices $M^T=[g_i m_{ji}]$ and $R=[r_{ij}]$ have the the same zero/nonzero pattern.
\end{enumerate}
\end{remark}
The human mobility model and  the virus transmission  dynamics is given by system \eqref{eq_t}.
\begin{subequations}
\label{eq_t}
 \begin{align}
 &\dfrac{d{S_H}_{ii}}{dt}=\!d_H({N_H}_i^r\!-\! {S_H}_{ii})\!-\!g_i {S_H}_{ii}\!+\!\sum_{k=1}^n r_{ik} {S_H}_{ik}\!-\!{{\beta}_H}_i\dfrac{Im_{i}}{{N_H}_i^p} {S_H}_{ii}\label{eq_t_a}\\
 & \dfrac{d{S_H}_{ij}}{dt}= \!g_i m_{ji} {S_H}_{ii}\!-\!d_H {S_H}_{ij}-r_{ij} {S_H}_{ij}\!-\! {{\beta}_H}_j \dfrac{Im_{j}}{{N_H}_j^p} {S_H}_{ij}\label{eq_t_b}\\
 &\dfrac{d{I_H}_{ii}}{dt}=\!- d_H {I_H}_{ii}\!-\!g_i {I_H}_{ii}\!+\!\sum_{k=1}^n r_{ik} {I_H}_{ik}\!+\!{{\beta}_H}_i\dfrac{Im_{i}}{{N_H}_i^p} {S_H}_{ii}\!-\! \gamma_H {I_H}_{ii}\label{eq_t_c}\\
  &\dfrac{d{I_H}_{ij}}{dt}=\!g_i m_{ji} {I_H}_{ii}\!-\! d_H {I_H}_{ij}\!-\!r_{ij} {I_H}_{ij}\!+\! {{\beta}_H}_j\dfrac{Im_{j}}{{N_H}_j^p} {S_H}_{ij}\!-\! \gamma_H {I_H}_{ij}\label{eq_t_d}\\
 &\dfrac{d{R_H}_{ii}}{dt}=\!\gamma_H {I_H}_{ii}\! -\! d_H {R_H}_{ii}-g_i {R_H}_{ii}\!+\!\sum_{k=1}^n r_{ik} {R_H}_{ik}\label{eq_t_e}\\
 &\dfrac{d{R_H}_{ij}}{dt}=\!g_i m_{ji} {R_H}_{ii}\!+\!\gamma_H {I_H}_{ij} \!-\! d_H {R_H}_{ij}\!-\! r_{ij} {R_H}_{ij} \label{eq_t_f}\\
 &\dfrac{d{S_m}_{i}}{dt}= \!s_L L_i^*\!-\! {d_m}{S_m}_{i}\!-\! \displaystyle\sum_{j=1}^n {\beta_m}_i \dfrac{{S_m}_{i}{I_H}_{ji}}{{N_H}_i^p}\label{eq_t_g}\\
 &\dfrac{d{I_m}_{i}}{dt}=\displaystyle\sum_{j=1}^n {\beta_m}_i \dfrac{{S_m}_{i}{I_H}_{ji}}{{N_H}_i^p}\! -\!{d_m} {I_m}_{i}\label{eq_t_h} 
 \end{align}
 \end{subequations}
where  $L_i^*=\left(1-\dfrac{1}{r}\right)\dfrac{{K_L}_i}{{\gamma_L}_i}$ and $\displaystyle {\gamma_L}_i=1+\dfrac{(s_L+d_L){K_L}_i}{s{K_E}_i}$ is given by the endemic equilibrium of the vector population dynamic model. 
In this model, each node is described by $(3n+2)n$ equations.

\subsubsection{Equilibrium of the model}

\begin{proposition}
\label{positivite}
The nonnegative orthant  $\mathbb{R_+}^{(3n+2)n}$ is positively invariant under the flow of \eqref{eq_t} and, for all $t>0$,  ${S_H}_{ii}>0$ and ${S_H}_{ij}>0$ provided that $g_i m_{ji}>0$.
Moreover, solutions of \eqref{eq_t} are bounded.
\end{proposition}

\begin{proof}
We easily see that solutions of system \eqref{eq_t} remain nonnegative.
Indeed, it is sufficient to show that for all nonnegative initial condition, the vector-field points out to the interior of the positive orthant. 
Assume now that  ${S_H}_{ii}=0$ at $t=0$, then  $ {d{S_H}_{ii}}/{dt}=d_H {N_H}_i^r+\sum_{k=1}^n r_{ik} {S_H}_{ik}>0,$ thus ${S_H}_{ii}>0$ for $t>0$.
Equally, if ${S_H}_{ij}=0$ at time $t=0$, then ${d{S_H}_{ij}}/{dt}= g_i m_{ji} {S_H}_{ii}>0.$

Finally, the boundedness follows from the positive invariance of $\mathbb{R_+}^{(3n+2)n}$ and the constant population property. Indeed
${d{N_H}_{i}^r}/{dt}=0,$ 
which means that, for any node $i$, the resident population is constant, thus by extension, the whole population is constant. Moreover, solutions of \eqref{eq_t} are bounded, since $(\mathbb{R^+})^{(3n+2)n}$ is invariant, human population is constant and vector population is bounded \cite{Moulay2010}.
\end{proof}

\begin{definition} \cite{Arino2009}.
\textit{1.} The system  is at an equilibrium if the time derivatives in \eqref{eq_t} are zero.\\
\textit{2.} A node $i$ is at the disease free equilibrium (DFE) if ${I_H}_{ji}=0$, ${I_m}_{i}=0$ for all $j=1,...,n$. \\
\textit{3.} The n-nodes model given by \eqref{eq_t} is at the DFE if each node is at the DFE, \textit{i.e.}, ${I_H}_{ji}=0$ and  ${I_m}_{i}=0$, for all $i,j=1,...,n$. 

\end{definition}

\begin{proposition}
System \eqref{eq_t} always has the following disease free equilibrium: 
\vspace*{-0.5cm}
\begin{eqnarray*}
{S_H}_{ii}^*&=&\left(\dfrac{1}{1+g_i\sum_{k=1}^n\frac{m_{ki}}{d_H+r_{ik}}}\right){N_H}_i^{r}, \ \ \ \ \ 
{S_H}_{ij}^*=g_i\dfrac{m_{ji}}{d_H+r_{ij}} {S_H}_{ii}^*\\
{I_H}_{ii}^*&=&0,\quad \quad  {I_H}_{ji}^*=0\\
{R_H}_{ii}^*&=&0,\quad \quad  {R_H}_{ji}^*=0\\
{S_m}_{i}^*&=& \dfrac{s_L}{d_m} L_i^*, \quad {I_m}_{i}^*=0,\quad \mbox{for all $i,j=1,...,n$, $i\neq j$.}
\end{eqnarray*}

\end{proposition}
\begin{proof}
It is sufficient to remark that \\$1-\sum_{k=1}^n \dfrac{m_{ki}r_{ik}}{d_H+r_{ik}}= \sum_{k=1}^n m_{ki}-\sum_{k=1}^n \dfrac{m_{ki}r_{ik}}{d_H+r_{ik}}=d_H \sum_{k=1}^n\frac{m_{ki}}{d_H+r_{ik}}.$  
 \end{proof}

\begin{theorem}\label{theogasDFE}\cite{Arino2009}.
Assume that system \eqref{eq_t} is at an equilibrium and a node $i$  is at the DFE.
Then all nodes that can be accessed from node $i$ and all nodes that have an access to node $i$ are at the DFE.
Moreover, if the outgoing matrix $M^T$ is irreducible, then all nodes are at the DFE.
\end{theorem}
\begin{proof}
See  appendix.
\end{proof}

\begin{definition}\cite{Arino2009}.
The disease is endemic within the population if the number of infective individuals is positive in this population.

The disease is endemic on node $i$ if there is a population on node $i$ in which the disease is endemic, \textit{i.e.}, there exists
$k\in\{1,...,n\}$ such as ${I_H}_{ki}>0$. 
\end{definition}

\begin{theorem}\label{theogasEnd} \cite{Arino2009}.
Assume that system \eqref{eq_t} is at an equilibrium and the disease is endemic on node $i$. 
Then the disease is endemic on all nodes that can be reached from node $i$.
In particular, if the matrix $M^T$ is irreducible, then the disease is endemic in all nodes.
\end{theorem}

\begin{proof}
See appendix.
%
\end{proof}

\subsection{The vector mobility model} 


Contrary to human displacements, it is unrealistic to identify mosquitoes  by their origin and destination. Nevertheless, we know that  \textit{Aedes albopictus} mosquitoes have a limited flight range.  Most mosquitoes disperse less than two hundred meters away from their original breeding place \cite{Vermillard,Nishida1993}. 
Indeed, mosquitoes present in a node $i$ may have an activity in neighboring nodes $j$, depending on their proximity (defined later).
Metapopulation models of the spread of vector-borne diseases (see for instance \cite{Zongo2008, Tsanou2011}) that include human long-distance displacements, do not take into account mosquito mobility. In our case of daily movements and very precise resolution, we cannot neglect the influence of mosquitoes displacements. We propose to model this activity using the biological radius of interaction around their breeding sites. Let us denote by $d_{ij}$ the distance between nodes $i$ and $j$ and by $d_{max}$ the maximum interaction radius of mosquitoes (approximately 200 m). In particular, we have $d_{ii}=0$ for all $i=1,...,n$. Now we assume that mosquitoes originated from node $i$ interact with population of node $j$, according to a function of the distance linearly decreasing. This function is given by 
\begin{eqnarray}\label{eq:distance_moustique}
\psi(d_{ij})=\left\{ \begin{array}{cl}\dfrac{d_{max}-d_{ij}}{d_{max}} &\mbox{ if } d_{ij}<d_{max}\\[3mm]
0 &\mbox{ else}\end{array}
\right.
\end{eqnarray}
Then, model \eqref{eq_t} becomes

\begin{subequations}
\label{eq_t2}
\begin{align}
\frac{d{S_H}_{ii}}{dt}&=\!d_H({N_H}_i^r\!-\! {S_H}_{ii})\!-\!g_i {S_H}_{ii}\!+\!\displaystyle\sum_{k=1}^n r_{ik} {S_H}_{ik}\!-\! \sum_{k=1}^n  {{\beta}_H}_i  \psi(d_{ik}) \dfrac{Im_{k}}{{N_H}_i^p} {S_H}_{ii}\!\label{eq_t_a2}\\
  \dfrac{d{S_H}_{ij}}{dt}&=\! g_i m_{ji} {S_H}_{ii}\!-\!d_H {S_H}_{ij}\!-\!r_{ij} {S_H}_{ij}\!-\!\displaystyle\sum_{k=1}^n {{\beta}_H}_j\psi(d_{ik}) \dfrac{Im_{k}}{{N_H}_j^p} {S_H}_{ij}\!\label{eq_t_b2}\\
  \dfrac{d{I_H}_{ii}}{dt}&=\!-\! d_H {I_H}_{ii}\!-\!g_i {I_H}_{ii}\!+\!\displaystyle\sum_{k=1}^n r_{ik} {I_H}_{ik}\!+\!\displaystyle\sum_{k=1}^n  {{\beta}_H}_i \psi(d_{ik})\dfrac{Im_{k}}{{N_H}_i^p} {S_H}_{ii}\!-\! \gamma_H {I_H}_{ii}\!\!\label{eq_t_c2}\\
  \dfrac{d{I_H}_{ij}}{dt}&=\!g_i m_{ji} {I_H}_{ii}\!-\! d_H {I_H}_{ij}\!-\!r_{ij} {I_H}_{ij}\!+\!\displaystyle\sum_{k=1}^n {{\beta}_H}_j\psi(d_{ik}) \dfrac{Im_{k}}{{N_H}_j^p} {S_H}_{ij}\!-\! \gamma_H {I_H}_{ij}\!\label{eq_t_d2}\\
  \dfrac{d{R_H}_{ii}}{dt}&=\!\gamma_H {I_H}_{ii}\! -\! d_H {R_H}_{ii}\!-\!g_i {R_H}_{ii}\!+\!\displaystyle\sum_{k=1}^n r_{ik} {R_H}_{ik}\!\label{eq_t_e2}\\
  \dfrac{d{R_H}_{ij}}{dt}&=\!g_i m_{ji} {R_H}_{ii}\!+\!\gamma_H {I_H}_{ij}\! -\! d_H {R_H}_{ij}\!-\! r_{ij} {R_H}_{ij}\!\label{eq_t_f2}\\
\dfrac{d{S_m}_{i}}{dt}&=\! s_L L_i\!-\! {d_m}{S_m}_{i}\!-\! \displaystyle\sum_{k=1}^n{\beta_m}_i\psi(d_{ik})\dfrac{{S_m}_{i}}{{N_H}_k^p} {I_H}_{k}^p\label{eq_t_g2}\\
\dfrac{d{I_m}_{i}}{dt}&=\! \displaystyle\sum_{k=1}^n{\beta_m}_i\psi(d_{ik})\dfrac{{S_m}_{i}}{{N_H}_k^p} {I_H}_{k}^p\!-\!{d_m} {I_m}_{i}\label{eq_t_h2}
 \end{align}
 \end{subequations}
 which describes both the vector population dynamic and the virus transmission, through the human and vector mobilities in the network.
Note that aquatic stages of eggs and larvae remain described by equations \eqref{eq_mous_c} and  \eqref{eq_mous_d} and complete the model.
\section{Application} 
\label{sec:application}

Willing to validate this approach, this section proposes a comparison of this model with a real-life example of chikungunya epidemic, namely, the event that occurred in 2005-2006 on the R\'eunion Island, Indian Ocean. This validation process makes a strong usage of real and realistic data, to reflect as much as possible the original environment. 
This last one is modeled based on real geographical information. Populations mobilityis modeled following  existing eal data analysis. Finnally, epidemiological results are compared to real data from the 2005-2006 event. 

\subsection{Distribution of the Human Population} 
\label{sub:distribution_of_the_human_population}

The R\'eunion Island is a mountainous region and the local population is not uniformly spread over the land.  To reflect this particular density in the metapopulation model, we rely on two realistic sources of information. 


The INSEE gives access to an estimated density of the population, based on a mesh zoning of the space \cite{INSEE}. This zoning consists in squares of one kilometer long. In each square, the number of people living in is estimated based on both tax cards and cadastral data. Figure \ref{fig:InseeMap} shows these $1~km^2$ squares with their associated population. 
\begin{figure}
	[ht!] \centering 
	\includegraphics[width=0.55
	\linewidth]{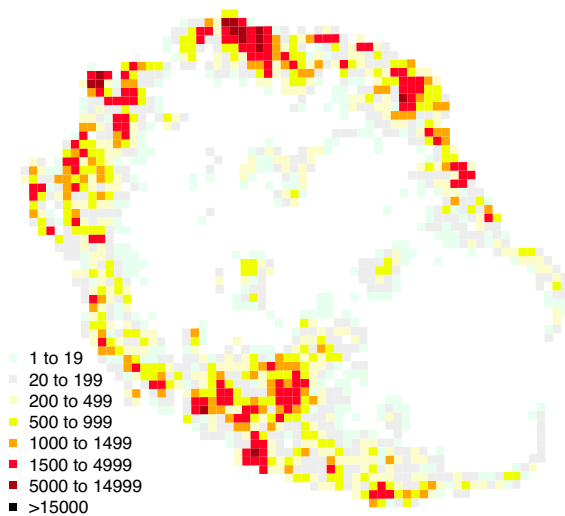} \caption{Estimated population distribution of the R\'eunion in 2007, according to the French Institution for Statistics, on a $1~km^2$ granularity.} 
	\label{fig:InseeMap} 
\end{figure}
  
There is a good confidence in the quality of this data, however, the limited granularity of one square kilometer is not precise enough for our purpose. 
Indeed this model envisions a lower scale of interaction between humans and mosquitoes, since the latter only have a couple hundreds of meters interaction range.

Secondly, the road network is used as a proxy to the human density. Indeed, the road network is denser where the population is itself dense, and lighter where no one lives. Data extracted from OpenStreetMap (OSM) \cite{OpenStreetMap} provide road information and especially road intersections at a high resolution (scale of one meter). We propose to use those intersections as the nodes of the metapopulation model. 

Then the human population is distributed on these nodes. Each of them belongs to one of the $1~km^2$-areas of the INSEE. So, on each of these squares, the local population can be evenly distributed among its own nodes. If a square has no node associated to it (it happens on low density areas) then a single node is created with the corresponding population.

Finally, this distribution is at least as relevant as the real data given by the INSEE at the scale of $1~km^2$, but the resolution goes below the $1~km$ limit, thanks to the crossroad/node analogy.  Figure \ref{fig:schema_construction} illustrates the main steps of the construction of the model. 

\begin{figure}
	[ht!] \centering \subfigure[]{
	\includegraphics[width=0.49
	\linewidth] {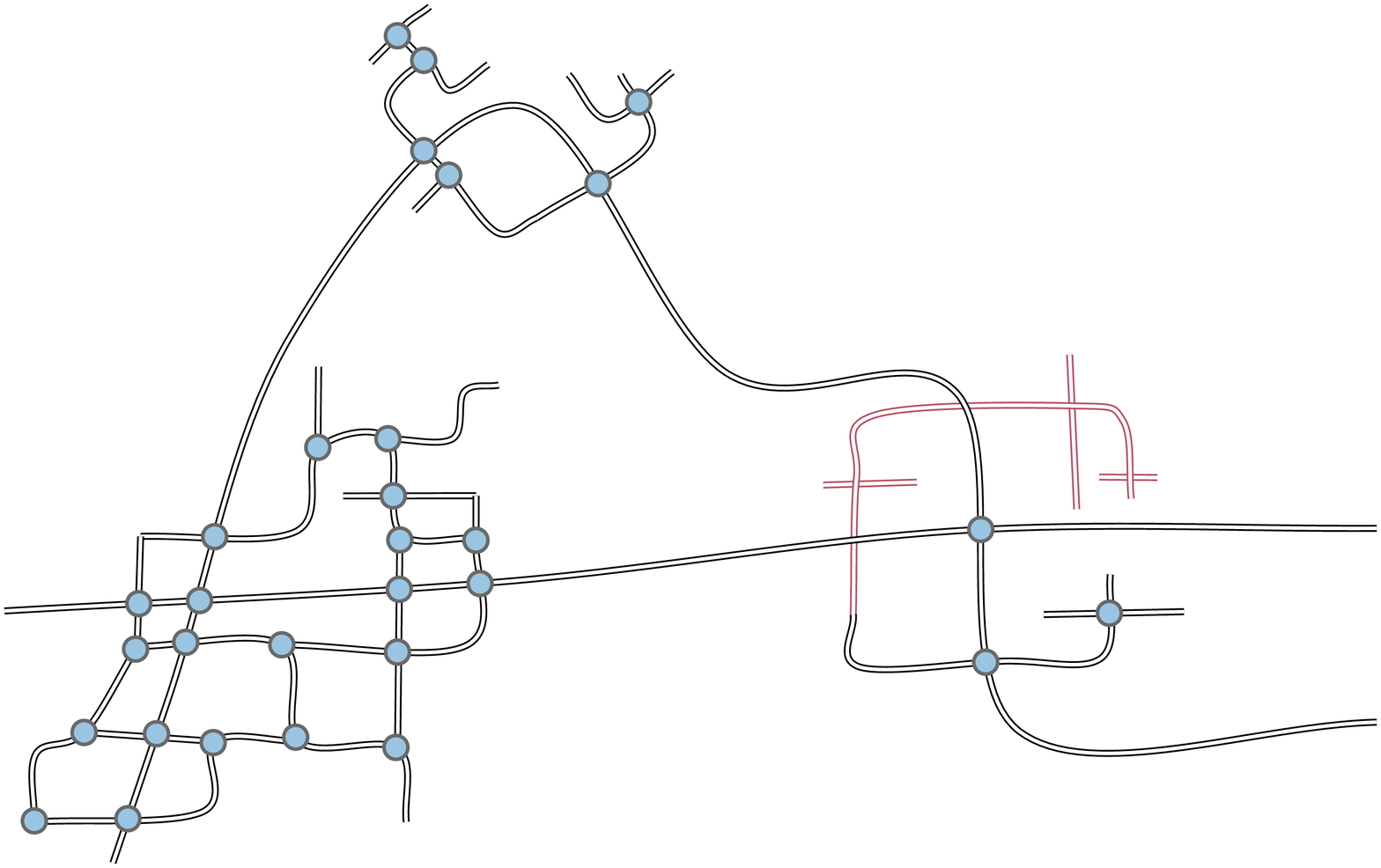}\label{sfig:roads2}} \subfigure[]{
	\includegraphics[width=0.48
	\linewidth] {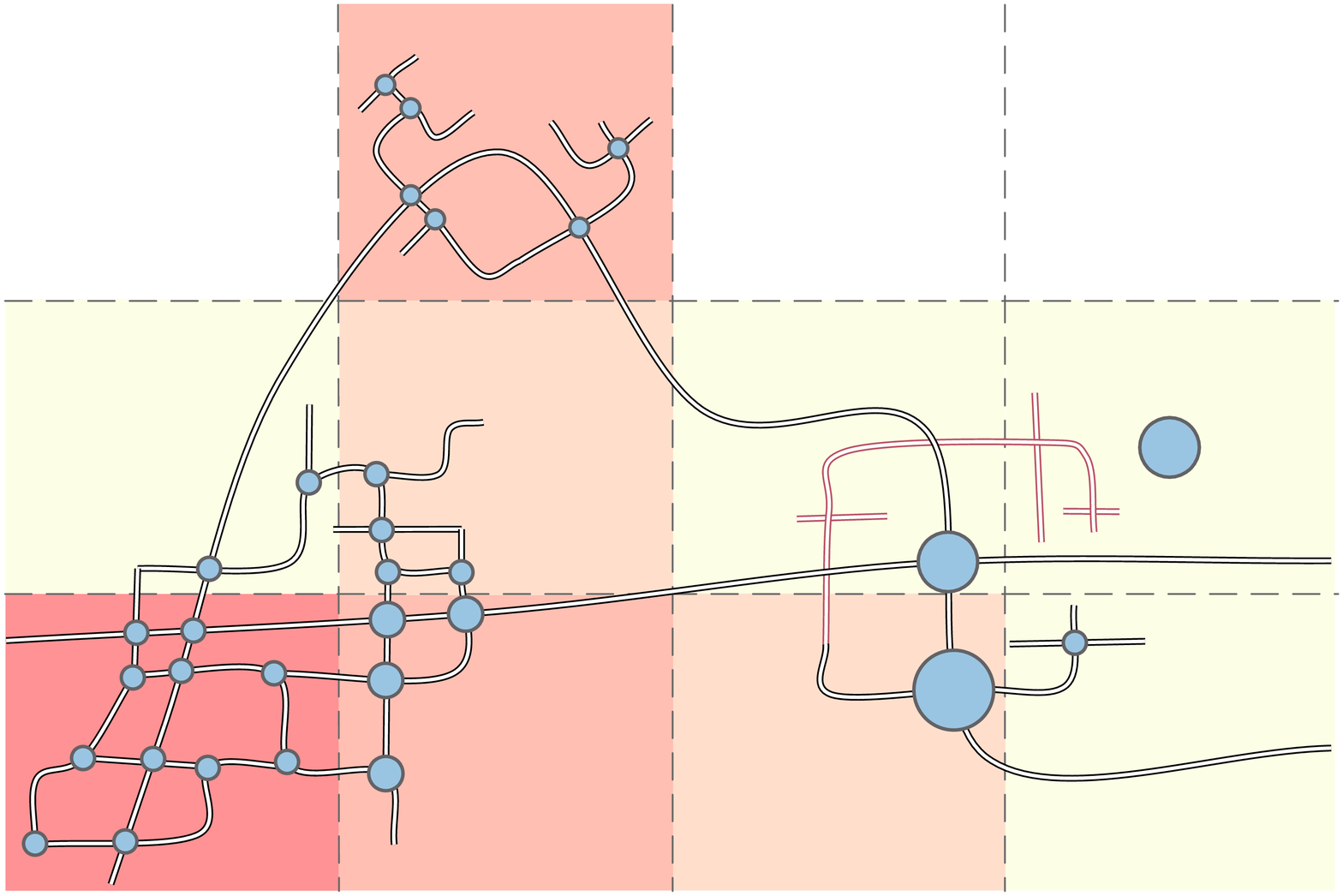}\label{sfig:roads3}} 
\caption[Metapopulation network schematic construction steps.]{
Schematic construction of the nodes of the metapopulation network. \ref{sfig:roads2}, network nodes are the roads intersections, knowing that this network may not be totally accurate (red roads missing in the map). \ref{sfig:roads3}, distribution of the population on nodes according to the INSEE data. Empty cells (with no known road network) are given one node (as in Fig. \ref{sfig:roads3}, second row, rightmost cell).}\label{fig:schema_construction} 
\end{figure}


\subsubsection{Distribution of the Mosquito Population} 
\label{ssub:distribution_of_the_mosquito_population}

No information about the density of mosquito population is available. Moreover, we are not interested in the whole population but rather  in the part of it that interacts with humans. Our hypothesis is that the density of mosquitoes is homogeneous. 

To retrieve this density from the human-centered network, each node is given a number of mosquitoes (through the carrying capacities) proportional to the geographical surface associated with that node, thanks to a simple Voronoï tessellation.

Since we only consider mosquitoes that interact with humans,  the surface is upper-bounded with the disk of radius $d_{max}$ that expresses the maximum interaction distance for a mosquito along its life. 

Let $S_i$ be the surface of node $i$. Let $S_{max}=\pi d_{max}^2$ be the surface associated with the maximal interaction radius $d_{max}$. Carrying capacities for each node are then set as follows:
$K_{E_i} = K_E \  \varphi(S_i)$ and $K_{L_i} = K_L\  \varphi(S_i)$, with $K_E$ and $K_L$,  constants for the carrying capacities of eggs, respectively larvae and $\varphi(S_i)$ is equal to $S_i / S_{max}$ and is upper bounded by $1$.


\subsubsection{Human Mobility} 
\label{ssub:human_mobility}

In the model, humans trips are given by  outgoing matrix $M^T$ and incoming matrix $R$ and represent for any pair  $i$ and $j$, the probability for humans living on $i$ to go to $j$. Non-null values in the matrix define edges between nodes in the metapopulation network. 

To reach realistic motion we need to rely on realistic data or model for human mobility.  Unfortunately, we could not get  such information  for  the Island, so we rely on some more general analysis work from  M. C. Gonz\'alez, C. A. Hidalgo and A. Barab\'asi \cite{Gonzalez:2008fk} who
analyze
mobile phone communication logs of 100,000 users. These logs register the geographical position of  mobile phones when calls happen or when texts are emitted/received. Mobile phones are a good proxy to the human mobility because users always carry their phone. 

The authors  were able to produce general laws on observed mobility patterns. These results give general formulas that describe the mobility observed. We use these formulas to generate new human mobility patterns. We consider the  three following  measures.
\paragraph{\textbf{Trips length}}The distribution of the length of a human trip $\Delta r$ in kilometer  according to the following power law: 
	\begin{equation}
		P(\Delta r)=(\Delta r + \Delta r_0)^{-\beta}\exp(-\Delta r / \kappa) 
	\end{equation}
	where $\Delta r_0=1,5 km$ is the cutoff value of the law and $\kappa= 80 km$.
 \paragraph{\textbf{Presence probability}}
 Given $N$ possible destinations for each individual, their presence probability in each destination is approximated with the following  Zipf law
$$f(k;N)=\frac{1/k}{\sum_{n=1}^N (1/n)},$$
where $k$ is the rank of each destination when decreasingly sorted. 
\paragraph{\textbf{Return probability}}
The authors observe that there is a peak of probability of returning to the same place every 24h hours. In other words, the human mobility mainly follows a daily pattern such as home/work trips. 

From these measures 
we
generate artificial per-individual trips on the island that respect the observed properties in the original data set.



\begin{remark}
The model defines  matrices $M^T$ and $R$ as tables indicating respectively departures and returns. The mobility model we could create out from  \cite{Gonzalez:2008fk} however gives the probability of presence for a human, given all of its possible destinations. Our hypothesis here is that there is a link between  departures / returns probabilities and presence probabilities that we do not investigate in this paper.
\end{remark}


\subsubsection{Mosquito Mobility} 
\label{ssub:mosquito_mobility}

As stated in the model description, mosquitoes  are not identified with origins and destinations. Since we are aware of their limited flight range, their mobility is thus defined by a geographical disk area of interaction centered at their origin node. Function \eqref{eq:distance_moustique} defines their interaction with humans. 

This new mobility pattern is constructed, in the metapopulation network, with  edges linking nodes that are below the interaction range $d_{max}$. 
%



\section{Results} 
\label{sec:results}

In this section we show and discuss results of the simulation of our metapopulation model applied to the scenario of the R\'eunion Island described in the previous section. Our results are then compared with real data from the 2005-2006 epidemic event.

\subsection{Analysis of the Metapopulation Network} 
\label{sub:analysis_of_the_metapopulation_network}

The metapopulation network can be represented by two graphs sharing the same set of nodes, or in other words, it is a graph with 2 subsets of edges. One subset is for human mobility and one is for mosquito mobility. 
The purpose of Table \ref{tab:graph_metrics} is to give some metrics for those two graphs to give a general overview of theirs dimension. Indeed due to its size no visualization of the all network is proposed since it does not give any useful information. 
\begin{table}[h]
	\caption{Metrics for the two graphs of the metapopulation network.}
	\centering
{\footnotesize
\begin{tabular}
	{lcc} \hline & Mosquitoes & Humans\\
	\hline number of nodes & 17988 & 17988 \\
	number of links & 151772 & 744313 \\
	average degree  & $\approx 17$ & $\approx 83$ \\
	connectivity  & no (1729 connected components) & yes \\
	diameter & 71 (for the biggest component) & 15 \\
	\hline   
\end{tabular}}
\label{tab:graph_metrics}
\end{table}

\subsection{Spread of the Disease in the Network} 
\label{sub:spread_of_the_disease_in_the_network}


In this scenario we observe the spread of the disease with  one infected individual put on one node of the network. Provided infection parameters ($\beta_H$ and $\beta_m$) are set high enough, the insertion of this individual in a system at disease-free equilibrium may start an epidemic event. Three nodes of the network are  monitored:  the first where the individual was inserted, a one-hop neighbor, and a farther one located $6 km$ from the insertion node. Four metrics are observed: susceptible humans ($S_H$), susceptible mosquitoes ($S_m$), infected humans ($I_H$), and infected mosquitoes ($I_m$).

Figure \ref{fig:3_nodes_neighbors} shows, as expected, a shift in time of the evolution is observed from the closest nodes to the farther one.
However, that shift (and thus the spread of the disease) is not proportional to the geographical distance between nodes but rather to a graph distance. The third observed node is 60 hops far from the insertion node in the mosquito mobility graph and only 4 hops away in the human mobility graph. Human mobility, thus,  greatly speeds up the spread.

\begin{figure}
	[ht!] \centering 
	\includegraphics[width=0.73\linewidth]{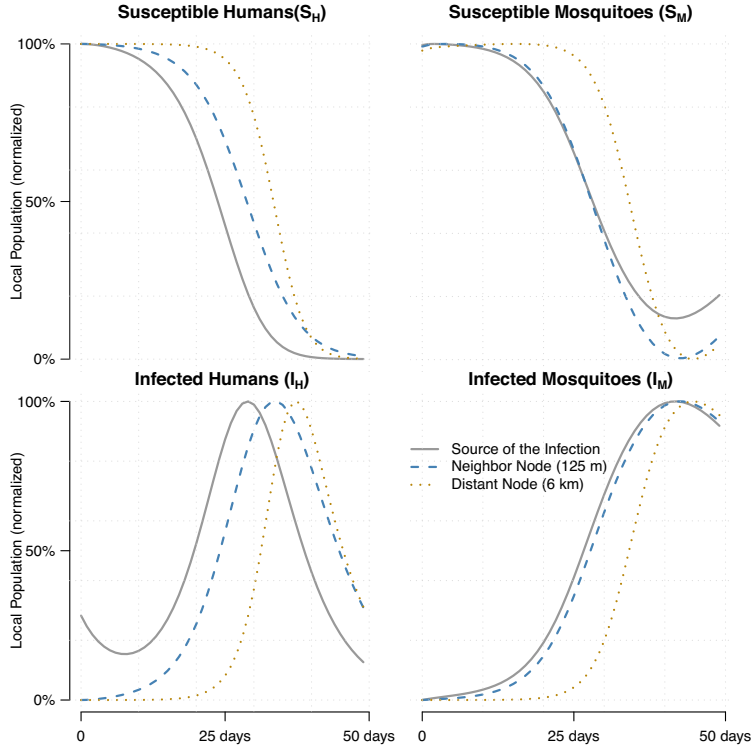} \caption[Spread of the disease in the network.]{Spread of the disease in the network for three observed nodes. The first  gets the infected human. The second  is an immediate neighbor ($125 m$ away from the first). The third, $6 km$ away from the first node, is 60 hops away in the mosquito mobility graph and  4 hops away in the human mobility graph. Quantities of population (y-axis) are normalized. A shift (or delay) effect in the spread is observed.}\label{fig:3_nodes_neighbors} 
\end{figure}


\subsection{Consequences of the Mosquito Mobility} 
\label{sub:consequences_of_the_mosquito_mobility}

As stated in the introduction, most of the metapopulation models on vector-borne viruses focus on long distance and long duration journeys for the human population. In those models the vector mobility can easily be ignored. Here we consider daily-based mobilitty with short journeys. We believe this shorter scale in time and space brings the new constraint that mosquito mobility cannot be neglected anymore.

The effect of local mosquito interaction is shown with a scenario where again, a disease free population is inserted an infected human.  Two experiments are carried out, one with mosquito mobility and one without it. Human mobility is kept in both scenarios. 

Figure \ref{fig:mosquito_mobility} shows the number of instantaneous infected humans ($I_H$). The infection starts rapidly and is concentrated with mosquito mobility enabled. Oppositely, without this mobility, the spread takes more time to start and the peak is less important. It is clear that mosquito mobility at this scale has to be considered.

\begin{figure}
	[ht!] \centering 
	\includegraphics[width=0.55\linewidth]{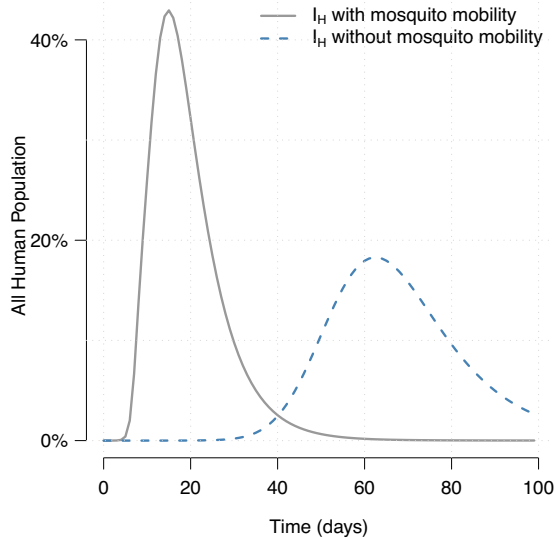} \caption[Consequences of the mosquito mobility.]{Effect of mosquito mobility of the spread of an epidemic. Evolution of the number of infected humans ($I_H$) with and without mosquito mobility.  In order to obtain comparable results, infection rates as set to a relatively high level: $\beta_H=0.2$  $\beta_m=0.15$. }\label{fig:mosquito_mobility} 
\end{figure}


\subsection{Consequences of the Human Mobility} 
\label{sub:consequences_of_the_human_mobility}
Human mobility has an obvious effect on the mobility. The purpose here is to analyze realistic and slight modifications of this mobility like quarantine measures that may be taken in case of epidemics.

This scenario proposes to stop the human mobility on infected parcels. So quarantine measures will be localized only on nodes where a given threshold of infection is reached. In quarantined nodes inhabitants are not allowed to move out and no foreigners are allowed in. Since mosquitoes  are not stopped by quarantine measures, they continue to interact within all nodes.

Figure \ref{fig:human_mobility} shows global instantaneous and cumulated infections values for humans at the scale of the all population. 
In this scenario, the disease would, without any control, reach 35\% of the population, just like the chikungunya event of 2005-2006. 

Results show the possibility to rapidly reduce the instantaneous infection rate. For instance, the $I_H$ peak  is almost divided by two when the quarantine threshold is set to  10\%. However looking at cumulated infection cases (seroprevalence), such a threshold does not significantly reduces the total amount of infected. To expect a real effect on the total number of infections one have to set the threshold below 1\%. Here, each node on average has less than 50 humans, so any threshold below 2\% is technically impossible to achieve.  

\begin{figure}
	[ht!] \centering 
	\includegraphics[width=0.65\linewidth]{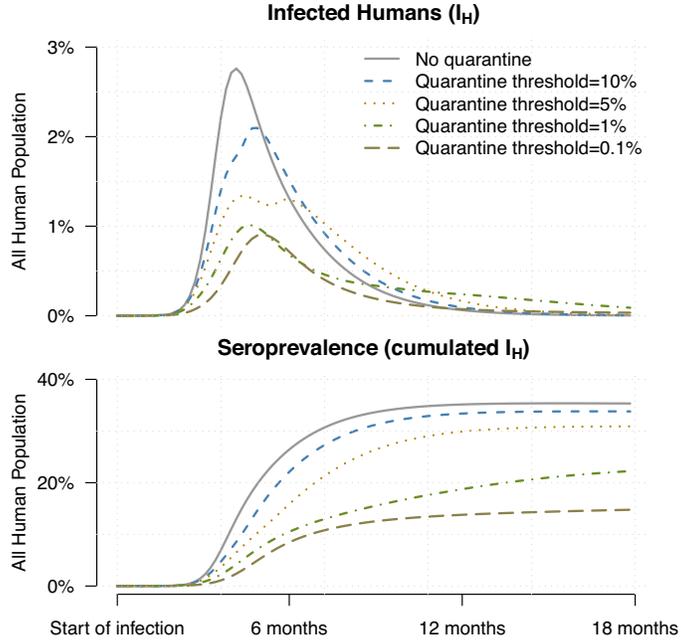} \caption[Consequences of the human mobility.]{Consequences of the human mobility. The instantaneous number of  infected humans and the seroprevalence (or number of people ever infected) are observed during an epidemic. Infection values depend on the quarantine threshold that reduce, node per node,  the human mobility, over the time. A threshold of  10\% indicates that a node with an infection rate of at least 10\% will be quarantined (incoming and outgoing human mobility are blocked).}\label{fig:human_mobility} 
\end{figure}


\subsection{Parameters Analysis} 
\label{sub:parameters_analysis}

The identification of predominant parameters of the system is needed to understand its behavior. We propose here a numerical  analysis of the parameters of the system. An in depth statistical sensitivity analysis or the analytical estimation of parameters would be of great interest too, but are out of the scope of this paper. 

Infection rates parameters, $\beta_H$ (mosquitoes to humans) and $\beta_m$ (humans to mosquitoes) have a strong impact on the results. 
Depending on the values selected for  $\beta_H$ and $\beta_m$, the infection goes from some few isolated cases to an epidemic that contaminates the whole population.


Fig. \ref{fig:betaH_betaM}  shows a scatter plot of values taken by $\beta_H$ and $\beta_m$. For each couple ($\beta_H$ ,$\beta_m$) the seroprevalence (after 400 days of epidemic) is observed. Values in percentage represent the ratio of the total human population ever infected depending on the two parameters. Contour lines help finding, given a seroprevalence percentage, values for $\beta_H$ and $\beta_m$. These vales are used in the validation. 

\begin{figure}
	[ht!] \centering 
	\includegraphics[width=0.5\linewidth]{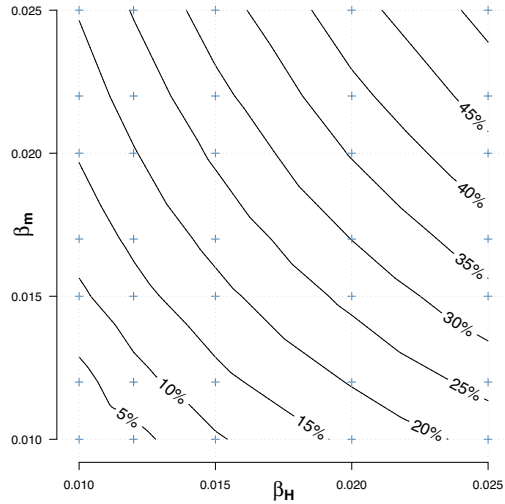} \caption{Effect of the  $\beta_H$ and $\beta_m$ parameters on the seroprevalence. Blue crosses are values obtained by simulation. Contour lines are obtained with a bilinear interpolation of these data points.}\label{fig:betaH_betaM} 
\end{figure}


\subsection{Validation against real data} 
\label{sub:validation}

We propose a validation with a comparison against real data from the 2005-2006 chikungunya epidemic. The two first cases of chikungunya where reported at the beginning of March 2005. The disease propagated during the following weeks and reached a peak at mid of May. The number of new cases then started to decrease until authorities thought the event was over at the end of the year 2005 with a seroprevalence (total number of cases) of 6,000 people. But in the second half of December 2005 the spread started over with a strength that was  not comparable to the previous peak. This second event reached a peak in February 2006 with more than 47,000 cases in a week. This sudden reactivation of the spread was later explained by a genetic mutation of the virus that lead to a new and more infective strain \cite{Vazeille2007}.  After the peak, the number of cases slowly started to decrease. The epidemic was declared over by April 2006. In the end, the InVS (French Institute for Health Care) counted 265,733 cases of chikungunya from March 2005 to April 2006. This represents more than 35\% of the total population of the Island. 

The following results are compared to real data indicating, week per week, new cases of the disease. This information was kindly provided by the InVS. 

Back to our model, the new virus strain with its stronger infection rate is modeled thanks to infection rates $\beta_H$ and $\beta_m$. We try to reproduce the all event by starting the simulation of the epidemic with one set of parameters, and then, by changing only one time the values of  $\beta_H$ and $\beta_m$ at the  moment it appears in the real event. The two sets of parameters where selected experimentally with the help of the previous study (see Fig. \ref{fig:betaH_betaM}) giving possible values for  $\beta_H$ and $\beta_m$ for a given seroprevalence. 


Figure \ref{fig:Reseau-mutation} compares the evolution of the real seroprevalence of the chikungunya event, with simulation results obtained with our model. Although the two curves do not fit perfectly, it allows a visual overall validation of the model. 

\begin{figure}
	[ht!] \centering 
	\includegraphics[width=0.65\linewidth]{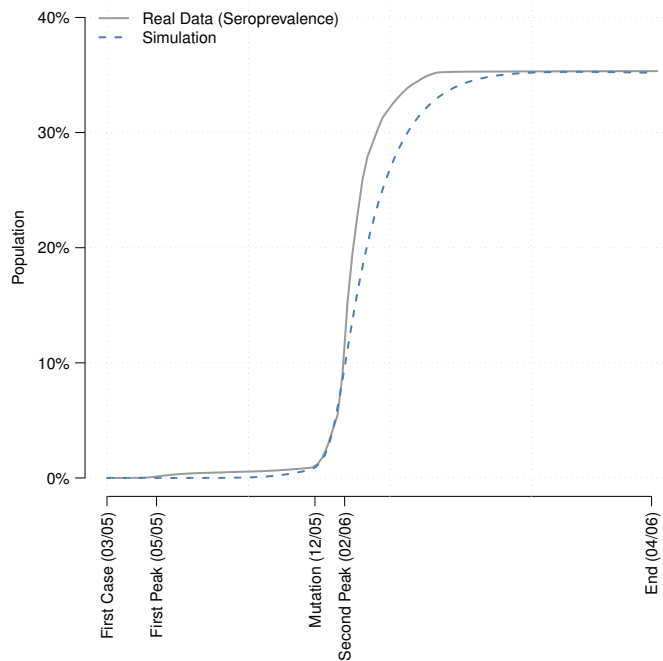}
	\caption[Comparison of the model with real data.]{Comparison of the model with real data. The seroprevalence is considered for the duration of the R\'eunion Island epidemic.  Values of   $\beta_H$ and $\beta_m$ are $0.0118$ and  $0.0101$,  before the mutation event and  $0.0245$ and $0.0161$, after. Those two couples $(\beta_H,\beta_m)$ are chosen thanks to the study in section \ref{sub:parameters_analysis}. Values for the other parameters of the system are given in Table \ref{tableau1}.}\label{fig:Reseau-mutation} 
\end{figure}

These results are a first step in the validation of this approach.    
In these experiments, only infection rates where investigated. Of course, other parameters need to be considered. Especially those leading the human mobility. Moreover the metapopulation approach needs to be validated at a lower scale than the scale of the all Island. These last observations form the perspectives of this work.


\subsection{Stochastic Analysis of the Model} 
\label{ssub:stochastic_analysis_of_the_model}

The model can be affected by stochastic changes. Indeed, the human mobility given by matrices $M^T$ and $R$ are defined according to the mobility probability law given in \cite{Gonzalez:2008fk}. 

We propose to study the robustness of the system by running simulations of the model with various instances of the human mobility matrices and thus observe the variations in the obtain results. 
The population itself and the number of displacements do not change, only the destinations are different.

Figure \ref{fig:nodes_Ih_analysis} illustrates the evolution of infected humans on various (randomly chosen) nodes in the network, within the very same scenario that the mutation experiment described in section 
\ref{sub:validation}.
We assume that observations made on $I_H$ also apply on the other outputs of the model. 30 different human mobility matrices are used. The figure shows that the height of the variation, although different on each node, remains within coherent boundaries. Indeed when considering, for each node, the date when the standard    deviation is the highest, compared to the node's population, it remains below 2\% of that  population, as illustrates Table  \ref{tab:nodes_Ih_analysis}. 

\begin{figure}[h]
	\centering
	\includegraphics[width=\linewidth]{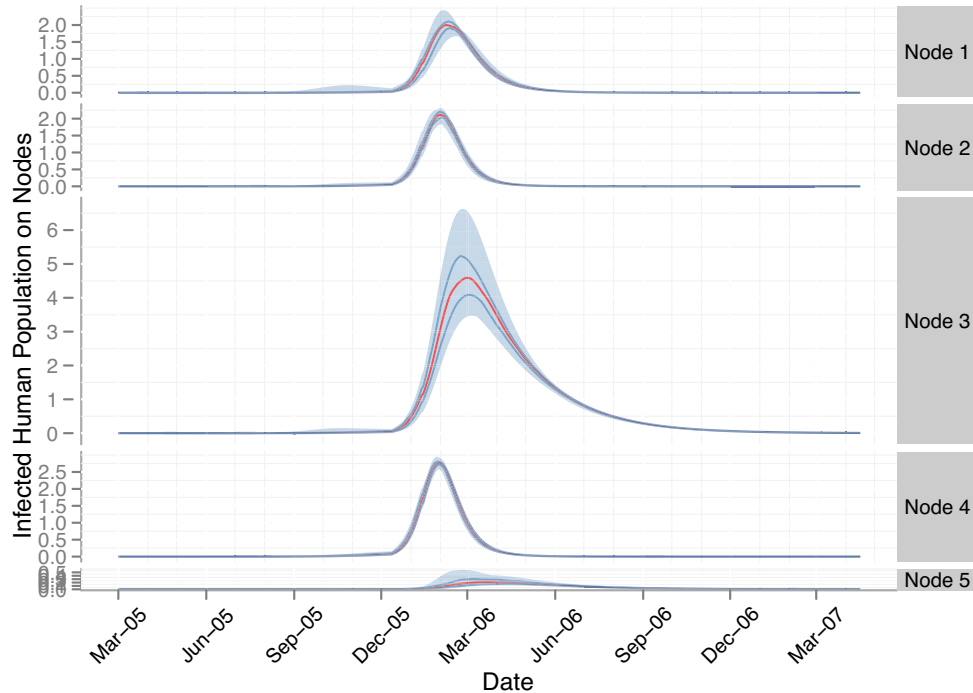}
	\caption{Statistics on $I_H$ for a random selection of nodes. The blue ribbon shows minimum and maximum values over the 30 experiments. Darker blue curves show $1^{st}$ and $3^{th}$ quartiles while the red curve shows the median.}\label{fig:nodes_Ih_analysis}
\end{figure}

\begin{table}[h]
	\caption{Highest values of Standard Deviation (SD) for $I_H$ on a random selection of nodes.}\label{tab:nodes_Ih_analysis}	
\centering{\footnotesize
\begin{tabular}{lccccc}
\hline
 &  Node 1 &   Node 2 &Node 3 &  Node 4 &   Node 5   \\
\hline
Population & 24 & 16 &  932 & 21 &  78\\
SD & 0.27 & 0.16 & 0.76 & 0.17  & 0.13 \\
date of SD (in 2006) &  01-24 & 01-16 & 02-15 & 01-12 & 02-21 \\
\% of population & 1.11 \% & 1.05 \% & 0.08 \% & 0.84 \% & 0.16 \% \\
\hline
\end{tabular}}
\end{table}
When considering simulation results at the scale of the whole network, the variation is even lower with a maximum standard deviation (on the $10^{th}$ of January 2016) of 544 infected humans which represents less that 0.07 \% of the population of the Island. 

These results tend to confirm the robustness of the system  against stochastic variations.



\section{Conclusion}\label{sec:conclusion}

In this work the main concern was to try to validate dynamical systems for the spread of a vectorial disease against a real epidemic scenario. We focused on the  R\'eunion Island's epidemic that occurred in 2005-2006. 

Many issues appear when trying to bridge the gap between global models and real life problems. Among them we chose to focus on the modeling of populations' mobilities. This logically implied a realistic modeling of environment, bringing us from temporal modeling to spatio-temporal modeling.

Designed to consider spatial interactions, the theory of metapopulations allows the special geographical environment of the Island to be included into the original model. This metapopulation network was created based on real geographical and demographical data. It could then handle local interactions between nodes and thus allow populations mobilities to become part of the model too. 

We proposed two mobility models. The first one, for humans, is based on the analysis of real human mobility data sets (mobile phones probes). The second model, for mosquitoes, is based on local interactions between nodes. 

After various studies of the consequences of mobilities on the spread of the disease, and after an analysis of the parameters of the system, a validation of the model against real seroprevalence data from the R\'eunion epidemic was proposed. Results validate the approach and clearly identify the human mobility as a key parameter in the spread of such an epidemic. 

As stated above, the choice in this work was to focus on mobilities. However, other recognized  issues play a key role in the all process. For instance, the effect of rain falls and weather seasons are known to directly influence mosquitoes evolution stages (especially aquatic phases). Studying the effect of seasons on the epidemic is definitely a perspective of this work. 

Finally, the integration of the metapopulation model and mobilities contributed to increase the complexity of the model. This last model which has only been explored on a numerical bases would need an analytical study.   

\section*{Acknowledgment}
The authors would like to thank the InVS and  particularly Ms. Elsa \textsc{Balleydier} epidemiologist at \emph{CIRE Oc\'ean Indien, InVS} in Saint Denis, R\'eunion Island, France,  for their help and availability.

\bibliographystyle{elsarticle-num}
\bibliography{biblio_reseau}
 
\cleardoublepage

\appendix
\section{}
\noindent Let  \hspace*{1cm}$V_{i\rightarrow}=\{ k\neq i \ : \ g_i m_{ki}>0\}  \mbox { and } V_{\rightarrow i}=\{ k\neq i \ / \ g_k m_{ik}>0\} $\\
sets of nodes that can be accessed \textit{directly} from city $i$ and nodes that have a \textit{direct} travel access to city i, respectively.
Let \vspace{-0,3cm}
\begin{align*}
&\!\mathcal{A}_{i\rightarrow}\!=\!\{ k\neq i  /  \!\exists (k_1, \dots,k_q)\mbox{ distincts} , \  \!m_{k_1 i}m_{k_2 k_1}\dots m_{k k_q}\!>\!0  \mbox{ and } g_ig_{k_1}\!\dots\! g_{k_q}\!>\!0\}\\
&\!\mathcal{A}_{\rightarrow i}\!=\!\{ k\neq i /  \! \exists (k_1, \dots,k_q) \mbox{ distincts} , \ \! m_{k_1 k}m_{k_2 k_1}\dots m_{i k_q}\!>\!0  \mbox{ and } g_kg_{k_1}\!\dots\! g_{k_q}\!>\!0\}
\end{align*}
sets of nodes that can be accessed from city i and  nodes that have an access to city i, respectively.

\paragraph{\textbf{Proof of theorem \ref{theogasDFE}}}
Assume that node $1$ is at the DFE (without loss of generality),
\textit{i.e.} ${I_H}_{k1}=0,$ $\forall k\!=\!1,...,n$ and ${I_m}_{1}=0$. From eq. \eqref{eq_t_c}, we have 
${d{I_H}_{11}}/{dt}=\sum_{k=1}^n r_{1k} {I_H}_{1k}.$
But node $i=1$ is at the DFE, \textit{i.e.}~$d{I_H}_{11}/dt=0$, and since $r_{1v}>0$,  $\forall v\in V_{1\rightarrow}$, then, ${I_H}_{1v}=0,$ $\forall v\in V_{1\rightarrow}$.
Consider now equation \eqref{eq_t_d} with $i=1$ and let $v\in V_{1\rightarrow}$, then $$\dfrac{d{I_H}_{1v}}{dt}={{\beta}_H}_v \dfrac{{I_m}_{v}}{{N_H}_v^p} {S_H}_{1v}.$$
As the system \eqref{eq_t} is at an equilibrium point, thus $d{I_H}_{1v}/{dt}=0$. However ${{\beta}_H}_1>0$ and ${S_H}_{1v}>0$ from proposition \ref{positivite}, then $Im_{v}=0$ for all $v\in V_{1\rightarrow}$.
Finally, we have to show that for all $v\in V_{1\rightarrow}$ and for all $k=1,...,n$, ${I_H}_{kv}=0$ \textit{i.e.}, all humans resident of node $k$ and present at node $v$  are not infected.
Consider equation \eqref{eq_t_h} for a node $v\in V_{1\rightarrow}$,then,
$$\dfrac{d {I_m}_{v}}{dt}=\displaystyle\sum_{j=1}^n {\beta_m}_v \dfrac{{S_m}_{v}{I_H}_{jv}}{{N_H}_i^p} -{d_m} {I_m}_{v}= {\beta_m}_v\dfrac{{S_m}_{v}}{{N_H}_i^p}\displaystyle\sum_{j=1}^n  {I_H}_{jv}={{\beta}_m}_v \dfrac{{S_m}_{v}}{{N_H}_v^p} {I_H}_v^p=0.$$
Since  ${{\beta}_m}_v >0$, and from proposition \ref{positivite}, ${S_m}_{v}>0$ then ${I_H}_v^p=0 $, \textit{i.e.}~$\sum_{k=1}^n {I_H}_{kv}=0$. It follows that for all $k=1,...,n$, ${I_H}_{kv}=0$, since  ${I_H}_{ij}\geq0$ for all $i,j=1,...,n$.
Thus, all adjacent nodes to node $i=1$ are the DFE equilibrium.
Moreover, by induction, we obtain that all nodes $v\in \mathcal{A}_{1\rightarrow}$ are at the DFE.

Assume now that $M^T$ is irreducible. From eq. \eqref{eq_t_d} with $j=1$, we have
${d{I_H}_{i1}}/{dt}= g_i m_{1i} {I_H}_{ii}.$
Since system is at the equilibrium $ d{I_H}_{i1}/{dt}=0$, beside $g_i m_{1i}>0$ for $i\in \mathcal{V}_{\rightarrow 1}$, then ${I_H}_{ii}=0$ for all $i\in \mathcal{V}_{\rightarrow 1}$. Let $v\in \mathcal{V}_{\rightarrow 1}$, from equation \eqref{eq_t_c}, we have
$$\dfrac{{dI_H}_{vv}}{dt}=\sum_{k=1}^n r_{vk} {I_H}_{vk}+{\beta_H}_i\dfrac{{I_m}_{v}}{{N_H}_v^p}{S_H}_{vv}=0,$$
then $\displaystyle\sum_{k=1}^n r_{vk} {I_H}_{vk}=0$ and ${\beta_H}_i\dfrac{{I_m}_{v}}{{N_H}_v^p}{S_H}_{vv}=0$ 
and   ${S_H}_{vv}>0,$  from proposition \ref{positivite}, thus ${I_m}_{v}=0$ for all $v\in \mathcal{V}_{\rightarrow 1}$.
Finally, consider equation \eqref{eq_t_h} for  $v\in \mathcal{V}_{\rightarrow 1}$, we have
$$\dfrac{d {I_m}_{v}}{dt}={{\beta}_m}_v \dfrac{{S_m}_{v}}{{N_H}_v^p} {I_H}_v^p=0,$$
then ${I_H}_v^p=0$, \textit{i.e.} $\displaystyle \sum_{k=1}^n {I_H}_{kv}=0$  $\forall  k=1,..., n$. Thus all nodes are at  the DFE.

%
%

\paragraph{\textbf{Proof of theorem \ref{theogasEnd}} }
\begin{proof}
Assume that the disease is endemic in node $i=1$, \textit{i.e.}  there exists $q\in{1,...,n}$ such that ${I_H}_{q1}>0$. 
We have to show that, in this case,  ${I_H}_{11}>0$.
If $q = 1$ then we can proceed. Assume that $q\neq 1$. 
Assume by contradiction  that ${I_H}_{11}=0$. Since the system is at the equilibrium, from \eqref{eq_t_c}, we have:
$$0= \dfrac{d{I_H}_{11}}{dt}=+\sum_{k=1}^n r_{1k} {I_H}_{1k}+{{\beta}_H}_1\dfrac{{I_m}_{1}}{{N_H}_1^p} {S_H}_{11}.$$
Beside ${\beta_H}_i>0 $ and ${S_H}_{11}>0$ then ${I_m}_{1}=0$. Thus, from equation \eqref{eq_t_h}, we have
$$ 0=\dfrac{d{I_m}_{1}}{dt}=\displaystyle\sum_{j=1}^n {\beta_m}_i \dfrac{{I_H}_{j1}}{{N_H}_i^p}{S_m}_{1}.$$
Beside ${\beta_m}_i>0 $ and ${S_m}_i>$ then ${I_H}_{j1}=0$ for all $ j=1,...n$, which is a contradiction. We have ${I_H}_{11}>0$, if the disease is endemic in node $1$, which we now assume.

Consider equation \eqref{eq_t_d} with $i=1$ and $j\neq i$. Assume ${I_H}_{ij}=0$. Since the system is at equilibrium, we have
$$ 0=\dfrac{d{I_H}_{ij}}{dt}=g_1 m_{j1} {I_H}_{11}+ {{\beta}_H}_j\dfrac{Im_{j}}{{N_H}_j^p} {S_H}_{ij}.$$
If $j \in \mathcal{V}_{1 \rightarrow }$ then $g_1 m_{j1}>0$ thus ${I_H}_{ii}=0$, which is a contradiction. Finally ${I_H}_{1j}>0$ for all $j \in \mathcal{V}_{1 \rightarrow }$, which means that the disease is endemic. Particularly, we deduce that ${I_H}_{jj}>0$ from the first part of the proof. By continuing, we can show that the disease is endemic in all nodes $j\in \mathcal{A}_{1 \rightarrow }$, that is to say, nodes reachable from node $1$.
\end{proof}

\end{document}